\documentclass[12pt]{amsart}
\usepackage{amssymb,amsxtra,amsmath,amsfonts,bm,hyperref}
\allowdisplaybreaks
\newcommand{\ds}{\displaystyle}
\newcommand{\p}[1]{\partial_{#1}}
\newcommand{\Z}{\mathbb{Z}}

\newcommand{\C}{\mathbb{C}}
\newcommand{\U}{\mathcal{U}}

\newcommand{\W}{\bar{W}}
\newcommand{\w}{\bar{w}}
\newcommand{\psibar}{\bar{\psi}}
\newcommand{\chibar}{\bar{\chi}}
\newcommand{\inv}{^{-1}}

\newcommand{\bt}{\bm{t}}
\newcommand{\bx}{\bm{x}}
\newcommand{\ba}{\bm{a}}
\newcommand{\bb}{\bm{b}}
\newcommand{\La}{\Lambda}
\newcommand{\dz}{\displaystyle\frac{dz}{2\pi\ii}}
\newtheorem{theorem}{Theorem}
\newtheorem{lemma}[theorem]{Lemma}
\newtheorem{proposition}[theorem]{Proposition}

\theoremstyle{definition}
\newtheorem{definition}[theorem]{Definition}

\theoremstyle{remark}
\newtheorem{remark}[theorem]{Remark}

\numberwithin{equation}{section}
\numberwithin{theorem}{section}

\newcommand\thref{Theorem \ref}
\newcommand\leref{Lemma \ref}
\newcommand\prref{Proposition \ref}

\newcommand\reref{Remark \ref}
\newcommand\seref{Section \ref}

\DeclareMathOperator\Wr{Wr}
\def\A{\mathcal{A}}

\def\ps{\partial_s}

\def\fin{\mathrm{fin}}
\def\x{\boldsymbol{x}}
\def\t{\boldsymbol{t}}
\def\bt{\bar{\t}}
\def\CC{\mathbb{C}}
\def\ii{\mathrm{i}} 

\def\la{\lambda}
\def\La{\Lambda}

\def\d{\partial}

\begin{document}

\title[Darboux transformations and Fay identities]{Darboux transformations and Fay identities for the extended bigraded Toda hierarchy}
\author{Bojko Bakalov}
\author{Anila Yadavalli}
\address{Department of Mathematics,
North Carolina State University,
Raleigh, NC 27695, USA}
\email{bojko\_bakalov@ncsu.edu; ayadava@ncsu.edu}

\thanks{The first author is supported in part by Simons Foundation grants 279074 and 584741}

\date{December 3, 2018}



\begin{abstract}
The extended bigraded Toda hierarchy (EBTH) is an integrable system satisfied by the Gromov--Witten total descendant potential of $\mathbb{CP}^1$ with two orbifold points. We write a bilinear equation for the tau-function of the EBTH and derive Fay identities from it. We show that the action of Darboux transformations on the tau-function is given by vertex operators. As a consequence, we obtain generalized Fay identities. 
\end{abstract}

\maketitle


\section{Introduction}\label{s0}

The \emph{extended Toda hierarchy} (ETH) was originally introduced by E. Getzler \cite{Ge01} and Y. Zhang \cite{Z02} in its bihamiltonian form, and later in its Lax form by G. Carlet, B. Dubrovin and Y. Zhang \cite{CDZ04}. It is obtained by adding an extra set of commuting flows to the 1D Toda hierarchy, which are given in terms of a ``logarithm" of the Lax operator. 
It was shown in \cite{Ge01,DZ04,Mil06,OP06} that the Gromov--Witten total descendant potential of $\mathbb{CP}^1$ is a tau-function of the ETH.

The \emph{extended bigraded Toda hierarchy} (EBTH) was introduced by G. Carlet \cite{Car06} as a generalization
of the extended Toda hierarchy related to the Frobenius manifolds from \cite{DZ98}. The EBTH is defined
for every pair $(k,m)$ of positive integers, and it coincides with the ETH for $k=m=1$. 
The total descendant potential of $\mathbb{CP}^1$
with two orbifold points of orders $k$ and $m$ is a tau-function of the EBTH (see \cite{MT08,CvdL13}).
The EBTH contains the bigraded Toda hierarchy, which is a reduction of the 2D Toda hierarchy
(see \cite{Tak10,UT84}). 

In this paper, we investigate how Darboux transformations of the EBTH affect the tau-function. Let us recall the B\"acklund--Darboux transformations and Fay identities 
of the KP hierarchy, following \cite{AvM94, BHY296, MZ97}. 
A \emph{B\"acklund--Darboux transformation} maps the Lax operator $L$ to a new pseudo-differential operator $\tilde{L} = A_{\Phi}LA_{\Phi}\inv$ and the wave function $\Psi$ to a new wave function $\tilde{\Psi} = A_{\Phi}\Psi$, where 
$A_{\Phi}\Psi = \Wr(\Phi, \Psi) / \Phi$ and $\Wr$ is the Wronskian. Then we have $\tilde{L}\tilde{\Psi} = z\tilde{\Psi}$, and if $\Psi$ is a wave function for the KP hierarchy, then $A_{\Psi(t,z_1)}\Psi$ is as well. In this case, the tau-function $\tilde{\tau}$ corresponding to the new solution is given by $\tilde{\tau}= X(t,z_1)\tau$, where 
\begin{equation*}
X(t,z) = \exp\Bigl(\sum_{j=1}^{\infty} t_jz^j\Bigr)\exp\Bigl(-\sum_{j=1}^{\infty} \frac{\partial_{t_j}}{j}z^{-j}\Bigr)
\end{equation*}
is the so-called vertex operator.
The proof of this theorem relies on the differential Fay identity, which is obtained by making a certain substitution in the bilinear equation for the KP hierarchy (see \cite[Lemma 2.1]{AvM94}). As a consequence of this result, one gets the generalized differential \emph{Fay identities} for the KP hierarchy:
\begin{align*}
\Wr &\bigl(\Psi(t,z_1), \dots, \Psi(t,z_N)\bigr) = \prod_{1 \leq i < j \leq N} (z_j - z_i ) \\
&\times \exp\Bigl({\sum_{i=1}^\infty t_i(z_1^{i} + \cdots + z_N^{i})}\Bigr) \, \frac{\tau(t-[z_1\inv] - \cdots - [z_N\inv])}{\tau(t)} \,,
\end{align*}
where $[z\inv] = (z\inv,z^{-2}/2,z^{-3}/3,\dots)$.
Conversely, K. Takasaki and T. Takebe \cite{TT95} showed that the Fay identities of the KP hierarchy 
imply the bilinear equation
of the hierarchy. L.P. Teo proved that the same is true for the Fay identities of the 2D Toda hierarchy \cite{Teo06}.

In this paper, we use the approach of K. Takasaki \cite{Tak10} to derive a \emph{bilinear equation} for the EBTH, which is 
equivalent to the one from \cite{CvdL13} after a change of variables. From this we obtain a difference Fay identity, similar to what was done in \cite{Tak08, Teo06} for the the 2D Toda hierarchy. Some Fay identities for the EBTH were given in \cite{LH10}, but writing the Fay identity in our notation allows us to study the action of Darboux transformations on the tau-function. In \cite{Car03}, G. Carlet defined \emph{Darboux transformations} on the wave functions for the ETH, and in \cite{LS16}, C. Li and T. Song generalized them to the EBTH. 

Our main result is that the action of Darboux transformations on the tau-function is given by applying the vertex operators 
\begin{align*}
\Gamma_{+}(z) &= e^{-\partial_s}\exp\Bigl(\ds\sum_{j=1}^\infty t_j z^j\Bigr)
\exp\Bigl(-\ds\sum_{j=1}^\infty \ds\frac{\partial_{t_j}}{j}z^{-j}\Bigr)
= e^{-\partial_s}X(t,z)
\end{align*}
and
\begin{equation*}
\Gamma_-(z) = z^se^{\p{s}}\exp\Bigl(-\ds\sum_{j=1}^\infty \bar{t}_j z^{-j}\Bigr)
\exp\Bigl(\sum_{j=1}^\infty\ds\frac{\partial_{\bar{t}_j}}{j}z^j\Bigr).
\end{equation*}	
Thus, new tau-functions for the EBTH can be obtained from existing ones by applying a product of $\Gamma_+$ and $\Gamma_-$ evaluated at different values $z_i \in \C^*$. As an application, we derive \emph{generalized Fay identities} for the EBTH. 

Now let us summarize the contents of the present paper.
In \seref{s1}, we start by reviewing difference operators and the extended bigraded Toda hierarchy (EBTH) following the approach of K. Takasaki \cite{Tak10}. Our version of the EBTH is related to the original definition of G. Carlet \cite{Car06} (or to \cite{CvdL13})
by an explicit change of variables, and we believe it is more convenient.
We discuss the Lax operator $L$, the wave operators $W$, $\bar W$, wave functions $\psi$, $\bar\psi$, and tau-function $\tau$ of the EBTH.

In \seref{s2}, we give an explicit bilinear equation for the EBTH that is equivalent to the one from \cite{CvdL13}, in the notation introduced by Takasaki. We provide a shorter proof than what was done in \cite{CvdL13}. From the bilinear equation written in this form, we get two difference Fay identities satisfied by the tau-function of the EBTH (cf.\ \cite{Tak08}). This is similar to what was done in \cite{LH10}, but we are following Takasaki's notation.

In \seref{s3}, we review the Darboux transformations on $L$ and $\psi$ from \cite{LS16}. We show that the action of the Darboux transformations on the tau-function is given by the vertex operators $\Gamma_+(z)$ and $\Gamma_-(z)$. This result is new even in the case $k = m = 1$ corresponding to the extended Toda hierarchy. We use it to conclude that new tau-functions can be found by acting on an existing tau-function $\tau$ with a product of $\Gamma_+(z_i)$ and $\Gamma_-(z_i)$ for certain $z_i \in \C^*$. 

In \seref{s4}, we apply a sequence of $N$ Darboux transformations and the vertex operators $\Gamma_+(z_i)$, $\Gamma_-(z_i)$ to derive generalized difference Fay identities for the EBTH.

Finally, \seref{s5} contains concluding remarks and open questions.

\section{Review of the extended bigraded Toda hierarchy}\label{s1}
This section is a quick review of the EBTH following \cite{BW16}. We first discuss the spaces of difference and differential-difference operators. Then we present a definition of the EBTH, its Lax operator, wave operators, wave functions, and tau-function.
\subsection{Spaces of difference and differential-difference operators}\label{sdiffo}
Consider functions of a variable $s$, and the shift operator $\La = e^{\p{s}}$ defined by $(\Lambda f)(s)	=	f(s+1)$.
The space $\A$ of (formal) \emph{difference operators} consists of all expressions of the form
\begin{equation*}
A=\sum_{i\in\Z}a_i(s)\Lambda^i.
\end{equation*}
We have $\A=\A_+\oplus\A_-$ where $\A_+$ (respectively, $\A_-$) consists of $A\in\A$ such that $a_i=0$ for all $i<0$
(respectively, $i\ge0$). For $A\in\A$, we define its projections
\begin{equation*}
A_+	=\sum_{i\geq 0}a_i(s)\Lambda^i \in\A_+, \qquad
A_-	=	\sum_{i<0}a_i(s)\Lambda^i \in\A_-.
\end{equation*}

We let $\A_{++}$ be the space of difference operators $A\in\A$ such that $a_i=0$ for $i\ll0$ (i.e., the powers of $\La$ are bounded from below),
and $\A_{--}$ be the space of $A\in\A$ such that $a_i=0$ for $i\gg0$ (i.e., the powers of $\La$ are bounded from above).
Both $\A_{++}$ and $\A_{--}$ are associative algebras, where the product is defined by linearity and
\begin{equation*}
(a(s)\La^i )(b(s)\La^j) = a(s)b(s+i) \La^{i+j}.
\end{equation*}
Let $\A_\fin=\A_{++}\cap\A_{--}$.
The product of a difference operator $A \in \A$ by an element of $\A_\fin$ is defined, but in general, the product of an element of $\A_{++}$ and an element of $\A_{--}$ is not well defined.

We will also consider the space $\A[\ps]$ of (formal) \emph{differential-difference operators}, where $\La\ps=\ps\La$.
Note that such operators depend polynomially on $\ps$.
Again, there is a splitting $\A[\ps]=\A_+[\ps]\oplus\A_-[\ps]$, and we have the associative algebras
$\A_{++}[\ps]$ and $\A_{--}[\ps]$, where the product is defined by linearity and
\begin{equation*}
(a(s)\La^i \ps^n)(b(s)\La^j \ps^m) = \sum_{k=0}^n \binom{n}{k} a(s) \frac{\partial^k b}{\partial s^k}(s+i) \La^{i+j} \ps^{m+n-k}.
\end{equation*}
Differential-difference operators can be applied to $z^s$ so that
\begin{equation*}
(a(s)\La^i \ps^n) z^s = a(s) z^i (\log z)^n  z^{s}.
\end{equation*}
\subsection{The extended bigraded Toda hierarchy}\label{sebth}
For fixed, positive integers $k$ and $m$,
consider a \emph{Lax operator} of the form
		\begin{equation*}
			L	=	\Lambda^k + u_{k-1}(s)\Lambda^{k-1} + \cdots + u_{- m}(s)\Lambda^{- m} \in\A_\fin, \qquad u_{- m}(s)\neq 0.
		\end{equation*}
There exist \emph{wave operators} (also called dressing operators):
\begin{equation}\label{wavop}
\begin{split}
W	&=	1 + \sum_{i=1}^\infty w_i(s)\Lambda^{- i} \in 1+\A_- \subset \A_{--}, \\
 \W	&=	\sum_{i=0}^\infty \bar{w}_i(s)\Lambda^i \in\A_+, \qquad \bar{w}_0(s)	\neq	0 ,
\end{split}
\end{equation}
such that
		\begin{equation}
			\label{BTHLaxdress}
			L	=	W\Lambda^kW\inv	=	\W\Lambda^{- m}\W\inv.
		\end{equation}
This allows us to define fractional powers of $L$ for any integer $n$ by
		\begin{equation}
			\label{EBTHfracL}
			L^{\frac{n}{k}}	=	W\Lambda^n W\inv	\in\A_{--},	\qquad	
L^{\frac{n}{m}}	=	\W\Lambda^{-n}\W\inv	\in\A_{++},
		\end{equation}
which commute with $L$ and satisfy
\begin{equation*}
\bigl(L^{\frac{n}{k}}\bigr)^k	=	\bigl( L^{\frac{n}{m}} \bigr)^m	=	L^n, \qquad n\in\Z_{\ge0}.
\end{equation*}
However, observe that $L^{\frac{n}{k}} \neq L^{\frac{p}{m}}$, unless $\frac{n}{k}=\frac{p}{m}\in\Z_{\ge0}$.
We define $\log L\in\A$ by
\begin{equation*}
			\log L	=	\frac{1}{2}W\p{s}W\inv - \frac{1}{2}\W\p{s}\W\inv	=	-\frac{1}{2}\frac{\partial W}{\partial s}W\inv + \frac{1}{2}\frac{\partial \W}{\partial s}\W\inv.
\end{equation*}
Then $\log L$ commutes with all $L^n$ for $n\in\Z_{\ge0}$,
but the composition of $\log L$ with a fractional power of $L$ is not well defined in general.

\begin{definition}[\cite{Car06}]\label{EBTHDef}
			The \emph{extended bigraded Toda hierarchy} (abbreviated EBTH) in Lax form is given by:
\begin{equation}\label{EBTHLaxflows}
			\begin{aligned}
\p{t_n} L &= [(L^{\frac{n}{k}})_+,L]	
, &	\qquad & n\geq 1, \\
\p{\bar t_n} L &= [(L^{\frac{n}{m}})_+,L]	
, &	\qquad & n\geq 1, \\
\p{x_n}	L &=	[(2L^n\log L)_+,L]	
,	&\qquad & n\geq 0.
			\end{aligned}
\end{equation}
		\end{definition}

The first two equations in \eqref{EBTHLaxflows} describe the \emph{bigraded Toda hierarchy},
which is a reduction of the 2D Toda hierarchy (see \cite{Tak10,UT84}). 
For $k=m=1$, the EBTH is equivalent to the \emph{extended Toda hierarchy} (ETH) \cite{CDZ04, Tak10}.

The flows of the EBTH induce flows on the dressing operators:
\begin{equation}\label{EBTHdresflows}
		\begin{aligned}
\p{t_n} W	&=	-(L^{\frac{n}{k}})_-W,	& \qquad
\p{t_n} \W &=	(L^{\frac{n}{k}})_+\W,\\
\p{\bar{t}_n} W	&=	-(L^{\frac{n}{m}})_-W,	& \qquad \p{\bar{t}_n} \W	&=	(L^{\frac{n}{m}})_+\W, \\
			\p{x_n} W	&=	-(2L^n\log L)_-W,	& \qquad \p{x_n} \W	&=	(2L^n\log L)_+\W.
		\end{aligned}
\end{equation}

\begin{remark}\label{rebth2}
Since $\d_{x_0}-\d_s$ and $\p{t_{nk}}-\p{\bar{t}_{nm}}$ act trivially on $L$, $W$ and $\W$, it follows that 
$L$, $W$ and $\W$ depend on $x_0+s$ and $t_{nk}+\bar{t}_{nm}$ for $n\geq 1$. Without loss of generality, we can assume 
$x_0=s$ and $t_{nk}=\bar{t}_{nm}$.
\end{remark}

\begin{remark}\label{rebth1}
To compare our version of the EBTH to the one from \cite{CvdL13}, we need to change there $\epsilon \mapsto-\epsilon$, which leads to 
$\Lambda \mapsto\Lambda\inv$ and $\zeta\mapsto\zeta\inv$ ($z$ here), and then apply the 
following change of variables:
		\begin{align*}
				x								&=	\epsilon s,\\
				q^{k-\alpha}_n	&=	\epsilon k\Bigl( n + \frac{\alpha}{k} \Bigr)_{n+1}t_{nk+\alpha}, \quad \alpha=1,2,\ldots,k-1,\\
				q^{k+\beta}_n		&=	\epsilon m\Bigl( n + \frac{\beta}{m} \Bigr)_{n+1}\bar{t}_{nm+\beta}, \quad \beta=1,2,\ldots,m-1,\\
				q^{k+m}_n				&=	\epsilon m(n+1)!\Bigl( 
t_{(n+1)k}+\bar t_{(n+1)m} + c_{n+1}\Bigl(\frac{1}{k} + \frac{1}{m}\Bigr)x_{k+1} \Bigr),\\
				q^k_n						&=	\epsilon n! \, x_n, \qquad n\geq 0.
		\end{align*}
Here $c_n$ are the harmonic numbers
\begin{equation*}
			c_0	=	0,	\quad c_k	=	1 + \frac{1}{2} + \frac{1}{3} + \cdots + \frac{1}{k}\,,
\end{equation*}
and 
$(p)_n$ denotes the Pochhammer symbol,
		\begin{align*}
			(p)_0								&=	1,\\
			(p)_n						&=	\prod_{i=1}^n (p-i+1), \qquad n\geq 1, \\
			(p)_{-n}	&=	\prod_{i=-n+1}^0 (p-i+1)\inv	=	\frac{1}{(p+n)_n} \,.
		\end{align*}
\end{remark}

Due to \reref{rebth2}, from now on we will always assume 
$x_0=s$ and $t_{nk}=\bar{t}_{nm}$ for $n\geq 1$. 
Introduce the notation
\begin{equation*}
\t=(t_1,t_2,\dots), \qquad \bt=(\bar t_1,\bar t_2,\dots), \qquad \x=(x_1,x_2,\dots),
\end{equation*}
and
\begin{equation*}
\xi(\t,z)=\sum_{i=1}^\infty t_iz^i, \qquad
\xi_k(\t,z)=\sum_{n=1}^\infty t_{nk}z^{nk} =\sum_{n=1}^\infty \bar{t}_{nm} z^{nk}.
\end{equation*}
Then
\begin{equation*}
\xi_m(\bt,z\inv)=\sum_{n=1}^\infty \bar{t}_{nm}z^{-nm} =\sum_{n=1}^\infty t_{nk} z^{-nm}.
\end{equation*}
We let
\begin{equation}\label{chi}
\begin{split}
			\chi	&=	z^{s + \xi(\x,z^k)}e^{\xi(\t,z) - \frac{1}{2}\xi_k(\t,z)},\\
			\bar{\chi}	&=	z^{s + \xi(\x,z^{-m})}e^{-\xi(\bt,z\inv) + \frac{1}{2}\xi_m(\bt,z\inv)}.
\end{split}
\end{equation}
Observe that, by definition,
\begin{equation}\label{chider}
\begin{split}
\p{t_i} \chi &= z^i \chi, \qquad \p{\bar{t}_j} \bar\chi = -z^{-j} \bar\chi \qquad\text{if}\quad k \nmid i, \;\; m \nmid j, \\
\p{t_{nk}} \chi &= \p{\bar{t}_{nm}} \chi = \frac12 z^{nk} \chi, \qquad
\p{t_{nk}} \bar\chi = \p{\bar{t}_{nm}} \bar\chi = -\frac12 z^{-nm} \bar\chi.
\end{split}
\end{equation}

The \emph{wave functions} $\psi$ and $\bar\psi$ of the EBTH are defined by:
\begin{equation}\label{EBTHwavef}
\begin{split}
				\psi	&=\psi(s,\t,\bt,\x,z)				=	W\chi	=	w\chi,\\
				\bar{\psi}	&=\bar\psi(s,\t,\bt,\x,z)	 =	\W\bar{\chi}	=	\bar{w}\bar{\chi},
\end{split}
\end{equation}
where
\begin{equation}\label{wavfn}
	\begin{split}
			w				
=	1	+	\sum_{i=1}^\infty w_i(s)z^{- i},\qquad
			\bar{w}	
=	\sum_{i=0}^\infty \bar{w}_i(s)z^i
\end{split}
\end{equation}
are the (left) symbols of $W$ and $\W$ respectively.	
Here we view $w$ and $\bar w$ as formal power series of $z\inv$ and $z$; however, in \seref{s3} below we will assume that $w(z)$ is convergent for $z$ in some domain $\U\subset\C$.

The wave functions satisfy
\begin{equation}\label{Lpsi}
L\psi=z^k\psi,  \qquad   L\bar\psi=z^{-m}\bar\psi.
\end{equation}
We have:
\begin{equation}\label{psiflow}
		\begin{aligned}
\p{t_n}{\psi} &=	(L^{\frac{n}{k}})_+{\psi}, & \qquad n&\in\Z_{\geq 1} \setminus k\Z,\\
\p{\bar t_n}{\psi} &= -(L^{\frac{n}{m}})_-{\psi}, & \qquad n&\in\Z_{\geq 1} \setminus m\Z,\\
\p{t_{nk}}{\psi} &= \p{\bar t_{nm}}{\psi} = A_n{\psi}, & \qquad n&\in\Z_{\geq 1}, \\
\p{x_n}{\psi} &=	(L^n\p{s} + P_n){\psi}, & \qquad n&\in\Z_{\geq 0},
		\end{aligned}
\end{equation}
and exactly the same equations hold for $\psibar$,
where 
\begin{equation}\label{An}
A_n =	\frac{1}{2}(L^n)_+ - \frac{1}{2}(L^n)_- 
					=	(L^n)_+ - \frac{1}{2}L^n\\
					=	\frac{1}{2}L^n - (L^n)_-
\end{equation}
and
\begin{equation}\label{Pn}
\begin{split}
				P_n	&=	-\Bigl(L^n\frac{\partial W}{\partial s}W\inv\Bigr)_+ - \Bigl(L^n\frac{\partial \W}{\partial s}\W\inv\Bigr)_-\\
						&=	L^nW\p{s}W\inv - (2L^n\log L)_- - L^n \p{s}\\
						&=	L^n\W\p{s}\W\inv + (2L^n\log L)_+ - L^n \p{s}.
\end{split}
\end{equation}
Observe that, due to \eqref{wavop} and \eqref{EBTHfracL}, we have
\begin{equation*}
(L^{\frac{n}{k}})_+ ,\; (L^{\frac{n}{m}})_- ,\; A_n,\; P_n \in\A_\fin, \qquad P_0=0.
\end{equation*}

Finally, by \cite{CvdL13, LH10}, there exists a \emph{tau-function} $\tau$ such that 
		\begin{align}
			\psi(s,\t,\bt,\bx,z) 
			&= \ds\frac{\tau(s,\t-[z\inv],\bt,\bx)}{\tau(s,\t,\bt,\bx)}\chi, \label{tau1}\\			
			\psibar(s,\t,\bt,\bx,z) 
			& = \ds\frac{\tau(s+1,\t,\bt+[z],\bx)}{\tau(s,\t,\bt,\bx)}\chibar \label{tau2},
		\end{align}
where
\begin{equation*}
[z\inv]=\Bigl(z\inv,\frac{z^{-2}}2,\frac{z^{-3}}3,\dots\Bigr), \qquad
[z]=\Bigl(z,\frac{z^2}2,\frac{z^3}3,\dots\Bigr).
\end{equation*}

\begin{remark}\label{rshift}
Since $t_{nk} = \bar{t}_{nm}$, we need to specify how to do the shifts $\t-[z\inv]$ in \eqref{tau1} and $\bt+[z]$ in \eqref{tau2}. Here and further, our convention is that in \eqref{tau1}, $\t-[z\inv]$ includes all variables $t_1,t_2,\dots$, while $\bt$ only includes 
$\bar{t}_i$ such that $m \nmid i$. Similarly, in \eqref{tau2}, all $\bar t_1,\bar t_2,\dots$ are shifted, while 
$\t$ only includes $t_i$ such that $k \nmid i$. 
\end{remark}

\section{Bilinear equation for the EBTH}\label{s2}
In this section, we derive a bilinear equation for the EBTH using Takasaki's approach from \cite{Tak10}. 
Our equation is equivalent to the bilinear equation from \cite{CvdL13}, and when $k = m =1$ it reduces to the 
bilinear equation for the ETH from \cite{Tak10}. As a consequence, we obtain
two difference Fay identities satisfied by tau-functions of the EBTH. 


\subsection{Dual wave functions}

Recall that the \emph{formal adjoint} of a difference operator $A = \sum_{i\in\Z} a_i(s)\Lambda^i \in \A$ is defined by 
\begin{equation*}
A^* = \sum_{i\in\Z} \Lambda^{-i} \circ a_i(s) 
= \sum_{i\in\Z} a_i(s-i) \Lambda^{-i}. 
\end{equation*}
It has the properties:
\begin{equation*}
(AB)^* = B^*A^*, \qquad (A^*)^*=A,  \qquad (A\inv)^* = (A^*)\inv. 
\end{equation*}

For given wave operators $W$ and $\bar W$, we define the \emph{dual wave functions} $\psi^*$ and $\bar\psi^*$ by:
\begin{equation}\label{wavef*}
	\begin{aligned}
		\psi^*&=(W^*)^{-1} \chi\inv = (W^*)^{-1} z^{-s-\xi(\bx,z^k)}e^{-\xi(\bt,z) + \frac12\xi_k(\t,z)},\\
		\bar{\psi}^* &= (\bar{W}^*)^{-1} \bar\chi\inv 
		= (\bar{W}^*)^{-1} z^{-s-\xi(\bx,z^{-m})}e^{\xi(\bt,z^{-1})-\frac12\xi_m(\bt,z^{-1})}.
	\end{aligned}
\end{equation}
If $W$ and $\bar W$ satisfy \eqref{BTHLaxdress}, \eqref{EBTHdresflows}, then it is easy to derive equations satisfied by $\psi^*$ and $\bar\psi^*$. For example, we have (cf.\ \eqref{Lpsi}):
\begin{equation}\label{Lpsi*}
L^*\psi^*=z^k\psi^*,  \qquad   L^*\bar\psi^*=z^{-m}\bar\psi^*.
\end{equation}
We will not list all the other equations, which are similar to \eqref{psiflow}, but we will need the following lemma.

\begin{lemma}\label{lpsi*eqs}
For every solution of the EBTH, the dual wave functions satisfy
\begin{align*}
\bigl( \p{x_n}-z^{nk}\p{s} \bigr) \psi^* &= -P_n^* \psi^*, \\
\bigl( \p{x_n}-z^{-nm}\p{s} \bigr) \bar\psi^* &= -P_n^* \bar\psi^*, 
\end{align*}
for all\/ $n \in \Z_{\geq 1}$, where\/ $P_n$ is given by \eqref{Pn}.
\end{lemma}
\begin{proof}
First, since $(\p{x_n}-z^{nk}\p{s})\chi = 0$, we have
\begin{align*}
\bigl( \p{x_n}-z^{nk}\p{s} \bigr) \psi^* = \chi\inv \bigl( \p{x_n}-z^{nk}\p{s} \bigr) (W^*)^{-1}.
\end{align*}
Using \eqref{EBTHdresflows}, we find
\begin{align*}
\p{x_n} W\inv &= - W\inv (\p{x_n} W) W\inv = W\inv (2L^n\log L)_-, \\
\p{s} W\inv &= \p{x_0} W\inv = W\inv (2\log L)_-.
\end{align*}
Note that taking formal adjoint commutes with taking derivative with respect to $x_n$, because the latter is done coefficient by coefficient. Hence,
\begin{align*}
\bigl( \p{x_n}-z^{nk}\p{s} \bigr) (W^*)^{-1} = \bigl( (2L^n\log L)_- - z^{nk} (2\log L)_- \bigr)^* (W^*)^{-1}.
\end{align*}
Then using \eqref{Lpsi*}, \eqref{Pn} and the fact that $P_0=0$, we obtain
\begin{align*}
\bigl( \p{x_n}-z^{nk}\p{s} \bigr) \psi^* = \bigl( (2L^n\log L)_- - L^n (2\log L)_- \bigr)^* \psi^*
= -P_n^* \psi^*.
\end{align*}
The second equation of the lemma is proved in the same way.
\end{proof}

\subsection{Bilinear equation for the wave functions}

The next result provides bilinear equations satisfied by the wave functions and dual wave functions of the EBTH.

\begin{theorem}\label{beq}
The wave functions\/ $\psi = W\chi$ and\/ $\psibar = \W\bar{\chi}$ solve the EBTH if and only if they satisfy the bilinear equation
\begin{align}
\notag
			\oint \ds\frac{dz}{2\pi\ii} z^{nk} \psi(s'-\xi(\ba,z^k)&, \t',\bt', \bx+\ba,z) \psi^*(s-\xi(\bb,z^k),\t,\bt,\bx+\bb,z)\\
\label{beebth}
			=\oint\ds\frac{dz}{2\pi\ii} z^{-nm} &\psibar(s'-\xi(\ba,z^{-m}), \t', \bt', \bx + \ba, z) \\
\notag
			\times\, &\psibar^*(s-\xi(\bb,z^{-m}),\t,\bt,\bx+\bb,z)
\end{align}
for all\/ $\ba=(a_1,a_2,\dots)$, $\bb=(b_1,b_2,\dots)$, 
$n \in \Z_{\geq 0}$ and\/ $s-s' \in \Z$.
\end{theorem}

\begin{remark}\label{taylor} 
By Taylor expansions of $\psi$ and $\psibar$ about $\t' = \t$, $\bt' = \bt$, the bilinear equation \eqref{beebth} is equivalent to:
\begin{align*}
\oint\ds \dz \bigl(\p{\t}^{\alpha}\p{\bt}^{\beta} \psi(s'-\xi(\ba,z^k)&,\t,\bt,\bx+\ba,z)\bigr)
\psi^*(s-\xi(\bb,z^k),\t,\bt,\bx+\bb,z) \\
= \oint\ds\dz \bigl(\p{\t}^{\alpha}\p{\bt}^{\beta} &\psibar(s'-\xi(\ba,z^{-m}),\t,\bt,\bx+\ba,z)\bigr) \\
\times\, &\psibar^*(s-\xi(\bb,z^{-m}),\t,\bt,\bx+\bb,z)
\end{align*}
for all multi-indices $\alpha$, $\beta$, where 
$\p{\t}^\alpha = \p{t_1}^{\alpha_1}\p{t_2}^{\alpha_2}\cdots$ and $\p{\bt}^\beta = \p{\bar{t}_1}^{\beta_1}\p{\bar{t}_2}^{\beta_2}\cdots$.
\end{remark}

\begin{remark}\label{zk} 
By taking a linear combination of equations \eqref{beebth} for different $n \in \Z_{\geq 0}$, 
we can replace $z^{nk}$ by $f(z^k)$ on the left side of \eqref{beebth} and $z^{-nm}$ by $f(z^{-m})$ on the right side,
for any formal power series $f(z)\in\C[[z]]$.
\end{remark}

The following lemma from \cite{Og08} will be useful in the proof of the above theorem.
In this lemma and below, we will use the notation $(A)_j = a_j(s)$ for the coefficient of $\La^j$ in 
a difference operator $A = \sum_{j\in\Z} a_j(s)\Lambda^j$.

	\begin{lemma}\label{ogawa}
	Let\/ $A$ and\/ $B$ be difference operators such that the product\/ $BA^*$ is well defined. Then 
		\begin{equation*}
		(BA^*)_j = \oint \ds\dz(\Lambda^jAz^s)(Bz^{-s}), \qquad j\in\Z.
		\end{equation*}
In particular, suppose that\/ $\bar A$, $\bar B$ are two other difference operators such that\/ $\bar B\bar A^*$ is well defined. Then
		\begin{equation*}
		\oint \ds\dz(\Lambda^jAz^s)(Bz^{-s}) = \oint \ds\dz(\Lambda^j\bar Az^s)(\bar Bz^{-s})
		\end{equation*}
for all\/ $j\in\Z$, if and only if\/ $BA^* = \bar B\bar A^*$.
	\end{lemma}
	
\begin{proof}[Proof of \thref{beq}]
First, following the approach of \cite{Tak10}, we will prove that the equations of the EBTH imply the bilinear equation \eqref{beebth}. 
By \eqref{EBTHwavef}, \eqref{wavef*} and \leref{ogawa}, we have
\begin{equation}\label{beebth1}
	\begin{aligned}
	\oint\dz & \bigl(\Lambda^j\psi(s,\t,\bt,\bx,z)\bigr) \psi^*(s,\t,\bt,\bx,z)\\
	 &= \oint\dz \bigl(\Lambda^j\psibar(s,\t,\bt,\bx,z)\bigr) \psibar^*(s,\t,\bt,\bx,z)
	\end{aligned}
\end{equation}	
for all $j\in\Z$. Therefore,
\begin{equation}\label{beebth2}
	\begin{aligned}
	\oint\dz &\psi(s',\t,\bt,\bx,z)\psi^*(s,\t,\bt,\bx,z) \\
	&= \oint\dz\psibar(s',\t,\bt,\bx,z)\psibar^*(s,\t,\bt,\bx,z)
	\end{aligned}
\end{equation}	
for all $s,s'$ with $s-s' \in \Z$.

Now applying $L^n$ as a difference operator with respect to $s'$ to both sides of \eqref{beebth2} and using \eqref{Lpsi}, we obtain 
\begin{equation}\label{beebth3}
	\begin{aligned}
	\oint\ds\dz & z^{nk}\psi(s',\t,\bt,\bx,z)\psi^*(s,\t,\bt,\bx,z)\\
	&= \oint\ds\dz z^{-nm}\psibar(s',\t,\bt,\bx,z)\psibar^*(s,\t,\bt,\bx,z)
	\end{aligned}
\end{equation}
for all $n \in \Z_{\geq 0}$ and $s-s' \in \Z$.
Recall that the action of the derivatives with respect to $\t$ and $\bt$ on the wave functions is given by difference operators
(see \eqref{psiflow}). 
We can apply the generating function
$\exp\bigl( \sum_{i = 1}^\infty c_i \p{t_i} \bigr)$
to $\psi$ and $\psibar$ in the above equation, thus shifting $\t$ by a constant $\bm{c}$. Let us denote $\t + \bm{c}$ by $\t'$. Doing the same for $\bt$, we get
	\begin{align}
\notag
	\oint\ds\dz & z^{nk}\psi(s',\t',\bt',\bx,z)\psi^*(s,\t,\bt,\bx,z) \\
\label{beqt}
	&= \oint\ds\dz z^{-nm}\psibar(s',\t',\bt',\bx,z)\psibar^*(s,\t,\bt,\bx,z).
\end{align}

Notice that, by \eqref{Lpsi} and \eqref{psiflow},
	\begin{align*}
\bigl( \p{x_\ell}-z^{\ell k}\p{s} \bigr) \psi &= Q_\ell \psi, \\
\bigl( \p{x_\ell}-z^{-\ell m}\p{s} \bigr) \psibar &= Q_\ell \psibar, \qquad 
Q_\ell = P_\ell - \ds\frac{\partial (L^\ell)}{\partial s} \,,
	\end{align*}
where $P_\ell$ is given by \eqref{Pn}.
We can apply the difference operator $Q_\ell$ to the variable $s'$ on both sides of \eqref{beqt} to obtain
\begin{align*}
	\oint\ds &\dz z^{nk} \bigl(\bigl( \p{x_\ell}-z^{\ell k}\p{s'} \bigr) \psi(s',\t',\bt',\bx,z) \bigr) \psi^*(s,\t,\bt,\bx,z) \\
	&= \oint\ds\dz z^{-nm} \bigl(\bigl( \p{x_\ell}-z^{-\ell m}\p{s'} \bigr)\psibar(s',\t',\bt',\bx,z) \bigr) \psibar^*(s,\t,\bt,\bx,z)
\end{align*}
for all $n\ge0$, $\ell\ge1$.
Using the generating function 
\begin{align*}
\exp \Bigl(\sum_{\ell=1}^\infty a_\ell \bigl( \p{x_\ell}-z^{\ell k}\p{s'} \bigr) \Bigr) \psi(s',\t',\bt',\bx,z)& \\
= \psi(s'-\xi(\ba,z^k), \t',\bt', \bx+\ba,z)&,
\end{align*}
we get
\begin{align*}
			\oint &\ds\frac{dz}{2\pi\ii} z^{nk} \psi(s'-\xi(\ba,z^k), \t',\bt', \bx+\ba,z)\psi^*(s,\t,\bt,\bx,z)\\
			&=\oint\ds\frac{dz}{2\pi\ii} z^{-nm}\psibar(s'-\xi(\ba,z^{-m}), \t', \bt', \bx + \ba, z)
			\psibar^*(s,\t,\bt,\bx,z).
\end{align*}
Similarly, by acting with $-P_\ell^*$ on $s$ in both sides of this equation and using \leref{lpsi*eqs}, we 
obtain the bilinear equation \eqref{beebth}.

Conversely, we have to prove that if $\psi$ and $\psibar$ satisfy the bilinear equation \eqref{beebth}, then they obey the equations of the EBTH. More precisely, suppose that the functions
	\begin{equation*}
		\psi = W\chi,  \quad  \psi^* = T\chi^{-1}, \quad
		\psibar = \bar W\chibar, \quad \psibar^* = \bar T\chibar^{-1}
	\end{equation*}
satisfy \eqref{beebth}, where $W$, $\bar W$, $T$, $\bar T$ are difference operators such that
\begin{equation*}
W, \, T^* \in 1+\A_-, \qquad \bar W, \, \bar T^* \in\A_+ 
\end{equation*}
(cf.\ \eqref{wavop}, \eqref{chi}, \eqref{EBTHwavef}, \eqref{wavef*}).
Then we will prove that $\psi$, $\psibar$ are the wave functions and $\psi^*$, $\psibar^*$ are the dual wave functions of a solution of the EBTH. 

First, setting $\ba = \bb = \boldsymbol 0$, $\t = \t'$, $\bt = \bt'$ in \eqref{beebth}, 
we obtain \eqref{beebth3} as a special case. Then putting $n = 0$ gives \eqref{beebth2}, and equivalently, \eqref{beebth1}. 
By Lemma \ref{ogawa}, equation \eqref{beebth1} implies that $TW^* = \bar T \bar W^*$.
Since $(TW^*)^*=WT^* \in 1+\A_{-}$ and $(\bar T \bar W^*)^* = \bar W \bar T^* \in \A_+$, we conclude that 
\begin{equation*}
T = (W^*)\inv, \qquad \bar T = (\bar W^*)\inv,
\end{equation*}
and \eqref{wavef*} holds.

Second, we define $L = W\La^{k}W\inv$ and want to prove \eqref{BTHLaxdress}. 
Notice that $L\psi = W\La^{k}W\inv W\chi = z^k\psi$. 
Applying $L$ with respect to $s'$ to both sides of \eqref{beebth2} and using \eqref{beebth3} for $n=1$, we get
	\begin{align*}
	\oint\ds\dz & \bigl( L \psibar(s',\t,\bt,\bx,z) \bigr) \psibar^*(s,\t,\bt,\bx,z)\\
	&= \oint\ds\dz z^{-m}\psibar(s',\t,\bt,\bx,z)\psibar^*(s,\t,\bt,\bx,z).
	\end{align*}
For $s'=s+j$ with $j\in\Z$, we have:
\begin{align*}
L \psibar(s',\t,\bt,\bx,z) &= \La^j L \bar W \bar\chi, \\
z^{-m}\psibar(s',\t,\bt,\bx,z) &= \La^j \bar W \La^{-m} \bar\chi.
\end{align*}
From Lemma \ref{ogawa}, it follows that 
\begin{equation*}
(\W^*)\inv(L\W)^* = (\W^*)^{-1}(\W\La^{-m})^*.
\end{equation*}
This simplifies to $L = \W\La^{-m}\W\inv$, thus proving \eqref{BTHLaxdress} and \eqref{Lpsi}. 

Next, we will show that we can identify $t_{nk}$ with $\bar{t}_{nm}$ in $L$, $W$ and $\bar W$ for $n \in \Z_{\geq 1}$ (cf.\ \reref{rebth2}).
Observe that, by \eqref{chider} and \eqref{EBTHwavef},
	\begin{align*}
	\frac{\partial\psi}{\partial {t_{nk}}} &= \ds\frac{\partial W}{\partial {t_{nk}}}\chi + \frac12z^{nk}W\chi, & \quad
	\frac{\partial\psibar}{\partial {t_{nk}}} &= \ds\frac{\partial\W}{\partial {t_{nk}}}\chibar - \frac12z^{-nm} \W\chibar, \\
	\frac{\partial\psi}{\partial {\bar{t}_{nm}}} &= \ds\frac{\partial W}{\partial {\bar{t}_{nm}}}\chi + \frac12z^{nk}W\chi, & \quad
	\frac{\partial\psibar}{\partial {\bar{t}_{nm}}} &= \ds\frac{\partial \W}{\partial {\bar{t}_{nm}}}\chibar - \frac12z^{-nm} \W\chibar;
	\end{align*}
hence,
	\begin{align*}
\Bigl(\ds\frac{\partial}{\partial {t_{nk}}} - \ds\frac{\partial}{\partial {\bar{t}_{nm}}} \Bigr) {\psi} &= \Bigl(\ds\frac{\partial {W}}{\partial {t_{nk}}} - \ds\frac{\partial {W}}{\partial {\bar{t}_{nm}}}\Bigr) {\chi}, \\
\Bigl(\ds\frac{\partial}{\partial {t_{nk}}} - \ds\frac{\partial}{\partial {\bar{t}_{nm}}} \Bigr) {\psibar} &= \Bigl(\ds\frac{\partial \bar{W}}{\partial {t_{nk}}} - \ds\frac{\partial \bar{W}}{\partial {\bar{t}_{nm}}}\Bigr) {\bar\chi}.
	\end{align*}
By \reref{taylor}, we can apply $\partial_{t_{nk}} - \partial_{\bar{t}_{nm}}$ to ${\psi}$ and $\psibar$ in the bilinear equation \eqref{beebth2} to obtain 
	\begin{align*}
	 \oint\dz &\Bigl(\Bigl(\ds\frac{\partial W}{\partial {t_{nk}}} - \ds\frac{\partial W}{\partial {\bar{t}_{nm}}}\Bigr)\chi(s')\Bigr)(W^*)\inv\chi^{-1}(s) \\
	 &= \oint\dz\Bigl(\Bigl(\ds\frac{\partial \W}{\partial {t_{nk}}} - \ds\frac{\partial \W}{\partial {\bar{t}_{nm}}}\Bigr)\chibar(s')\Bigr)(\bar W^*)\inv\chibar^{-1}(s)
	\end{align*}
for $s-s'\in\Z$.
Using \leref{ogawa} as before, we get
	\begin{equation*}
		(W^*)\inv\left(\ds\frac{\partial W}{\partial {t_{nk}}} - \ds\frac{\partial W}{\partial {\bar{t}_{nm}}}\right)^*
		= (\W^*)\inv\left(\ds\frac{\partial \W}{\partial {t_{nk}}} - \ds\frac{\partial \W}{\partial {\bar{t}_{nm}}}\right)^*,
	\end{equation*}	
or equivalently,
	\begin{equation*}
		\left(\ds\frac{\partial W}{\partial {t_{nk}}} - \ds\frac{\partial W}{\partial {\bar{t}_{nm}}}\right) W\inv
		= \left(\ds\frac{\partial \W}{\partial {t_{nk}}} - \ds\frac{\partial \W}{\partial {\bar{t}_{nm}}}\right) \bar W\inv.
	\end{equation*}
By \eqref{wavop}, the left-hand side of this equation lies in $\A_-$, while the right-hand side in $\A_+$. Therefore, both sides  vanish.

To finish the proof of the theorem, it is left to show that if $\psi$ and $\psibar$ satisfy the bilinear equation \eqref{beebth}, then they satisfy \eqref{psiflow}. First, consider the derivatives with respect to $t_{nk}$ and $\bar{t}_{nm}$ for $n \in \Z_{\geq 1}$.
As above, we have
\begin{align*}
\frac{\partial\psi}{\partial {t_{nk}}} = \ds\frac{\partial W}{\partial {t_{nk}}}\chi +	 \frac12z^{nk}W\chi
= \ds\frac{\partial W}{\partial {t_{nk}}}\chi +	 \frac12 L^{n}W\chi,
\end{align*}
which implies
\begin{align*}
\Bigl(\ds\frac{\partial}{\partial {t_{nk}}}-A_n\Bigr)\psi 
=\Bigl(\ds\frac{\partial W}{\partial {t_{nk}}} + (L^n)_-W\Bigr)\chi,
\end{align*}
where $A_n$ is given by \eqref{An}. Similarly,
\begin{equation*}
\Bigl(\ds\frac{\partial}{\partial {t_{nk}}} - A_n\Bigr)\psibar 
= \Bigl(\ds\frac{\partial\W}{\partial {t_{nk}}}-(L^n)_+\W\Bigr)\chibar.
\end{equation*}
We can apply the operator $\partial_{t_{nk}}-A_n$ to $\psi$ and $\psibar$ in the bilinear equation \eqref{beebth2}. By \leref{ogawa} again, we obtain
	\begin{align*}
		\ds\frac{\partial W}{\partial {t_{nk}}} = \ds\frac{\partial W}{\partial {\bar{t}_{nm}}} = -(L^n)_-W, \qquad
		\ds\frac{\partial \W}{\partial {t_{nk}}} = \ds\frac{\partial \W}{\partial {\bar{t}_{nm}}} = (L^n)_+W,
	\end{align*}
as claimed.

Next, let $n$ be such that $k$ does not divide $n$. Using \eqref{Lpsi}, we get
\begin{align*}
\Bigl(\ds\frac{\partial}{\partial {t_n}} - (L^{\frac{n}{k}})_+\Bigr)\psi 
&=\Bigl(\ds\frac{\partial W}{\partial {t_n}}\chi + z^nW\chi - (L^{\frac{n}{k}})_+W\chi\Bigr)\\
&=\Bigl(\ds\frac{\partial W}{\partial {t_n}} + (L^{\frac{n}{k}})_-W\Bigr)\chi,
\end{align*}
and similarly,
\begin{align*}
\Bigl(\ds\frac{\partial}{\partial {t_n}} - (L^{\frac{n}{k}})_+\Bigr)\psibar 
=\Bigl(\ds\frac{\partial \W}{\partial {t_n}}-(L^\frac{n}{k})_+\W\Bigr)\chibar.
\end{align*}
Applying the operator $\partial_{t_n} - (L^\frac{n}{k})_+$ to $\psi$ and $\psibar$ in \eqref{beebth2} and using \leref{ogawa} gives 
	\begin{equation*}
		\ds\frac{\partial W}{\partial {t_n}} + (L^{\frac{n}{k}})_-W = \ds\frac{\partial \W}{\partial {t_n}}-(L^\frac{n}{k})_+\W = 0.
	\end{equation*}

Finally, consider the derivatives with respect to the logarithmic variables $x_n$. 
By \eqref{Pn} and $\psi = W\chi$, we see that
	\begin{align*}
	\Bigl(\ds\frac{\partial}{\partial x_n} &- (L^n\partial_s + P_n)\Bigr)\psi\\ 
	&= \ds\frac{\partial W}{\partial x_n}\chi + W\ds\frac{\partial\chi}{\partial x_n} 
	- \bigl(L^nW\partial_sW\inv - (2L^n\log L)_- \bigr)\psi \\
	&= \ds\frac{\partial W}{\partial x_n}\chi + z^{nk}\log(z)W\chi - z^{nk}\log(z)W\chi + (2L^n\log L)_-W\chi\\
	&=\Bigl(\ds\frac{\partial W}{\partial x_n} + (2L^n\log L)_-W\Bigr)\chi.
	\end{align*}
Similarly,
\begin{equation*}
\Bigl(\ds\frac{\partial}{\partial x_n} - (L^n\partial_s + P_n)\Bigr)\psibar 
= \Bigl(\ds\frac{\partial \W}{\partial x_n} - (2L^n\log L)_+\W\Bigr)\chibar.
\end{equation*}
Applying the operator $\partial_{x_n} - (L^n\partial_s + P_n)$ to $\psi$ and $\psibar$ in \eqref{beebth2} gives 
\begin{align*}
\oint\dz &\Bigl(\Bigl(\frac{\partial W}{\partial x_n} + (2L^n\log L)_-W\Bigr)\chi(s')\Bigr) (W^*)\inv\chi\inv(s) \\ 
&= \oint\dz\Bigl(\Bigl(\frac{\partial \W}{\partial x_n} - (2L^n\log L)_+\Bigr)\chibar(s')\Bigr) (\W^*)\inv\chibar\inv(s).
\end{align*}
By \leref{ogawa}, this implies
\begin{equation*}
\left(\ds\frac{\partial W}{\partial x_n} + (2L^n\log L)_-W\right) W\inv
= \left(\ds\frac{\partial \W}{\partial x_n} - (2L^n\log L)_+\W\right)\bar W\inv.
\end{equation*}
Since the left side is in $\A_-$ and the right side is in $\A_+$, both sides must vanish.
This completes the proof of \thref{beq}. 
\end{proof}

\subsection{Bilinear equation for the tau-function}

In this subsection, we will derive a bilinear equation satisfied by the tau-function $\tau$ of the EBTH. Recall that the wave functions $\psi$ and $\psibar$ can be expressed in terms of $\tau$ by \eqref{tau1}, \eqref{tau2}. Next, we do it for the dual wave functions defined by \eqref{wavef*}.

\begin{proposition}\label{stars}
The dual wave functions\/ $\psi^*$ and\/ $\psibar^*$ of the EBTH can be expressed in terms of the tau-function\/ $\tau$ as follows$:$
\begin{align}
\label{tau3}
			\psi^*(s,\t,\bt,\bx,z) &= \ds\frac{\tau(s,\t+[z\inv],\bt,\bx)}{\tau(s,\t,\bt,\bx)}\chi\inv ,\\
\label{tau4}
		     \psibar^*(s,\t,\bt,\bx,z) &= \ds\frac{\tau(s-1,\t,\bt-[z],\bx)}{\tau(s,\t,\bt,\bx)}\chibar\inv,
\end{align}
where we use the convention of \reref{rshift}.
\end{proposition}
\begin{proof}
Let us write
\begin{equation*}
\psi=w\chi, \qquad \bar\psi=\bar w\bar\chi, \qquad \psi^*=w^*\chi\inv, \qquad \bar\psi^*=\bar w^*\bar\chi\inv, 
\end{equation*}
for some functions $w,\bar w,w^*, \bar w^*$ (cf.\ \eqref{EBTHwavef}, \eqref{wavef*}).
Setting $s'=s$, $\ba=\bb=\boldsymbol 0$ 
in the bilinear equation \eqref{beebth}, we get 
	\begin{align*}
		\oint &\dz z^{nk} e^{\xi(\t'-\t,z) - \frac12\xi_k(\t'-\t,z)} w(s',\t',\bt',\bx,z)w^*(s,\t,\bt,\bx, z) \\
		&= \oint\dz z^{-nm} e^{-\xi(\bt' - \bt, z\inv) + \frac12\xi_m(\bt' - \bt,z\inv)}
		\w(s',\t', \bt', \bx,z)\w^*(s,\t,\bt,\bx,z).
	\end{align*}
According to \reref{zk}, we can replace $z^{nk}$ in the left-hand side by $f(z^k)$, and $z^{-nm}$ in the right-hand side by $f(z^{-m})$, for any $f(z)\in\C[[z]]$. If we do it for
\begin{equation*}
 f(z^k) = e^{\frac12\xi_k(\t'-\t,z)} = \exp\frac12\sum_{n=1}^\infty(t'_{nk} - t_{nk})z^{nk}
 = \exp\frac12\sum_{n=1}^\infty(\bar{t}'_{nm} - \bar{t}_{nm})z^{nk}, 
 \end{equation*}
then $f(z^{-m})=e^{\frac12\xi_m(\bt' - \bt,z\inv)}$, and we obtain
\begin{align*}\label{fzk}
		\oint &\dz e^{\xi(\t'-\t,z)}w(s',\t',\bt',\bx,z)w^*(s,\t,\bt,\bx, z) \\
		&= \oint\dz \exp\Bigl({\ds-\sum_{m\nmid i} ({\bar t}'_i -\bar{t}_i)z^{-i}}\Bigr)\w(s',\t', \bt', \bx,z)\w^*(s,\t,\bt,\bx,z).
	\end{align*}
Now setting $\bar{t}_{i}'=\bar{t}_{i}$ for $m \nmid i$, $\t' = \t+[u\inv]$ and using 
\begin{equation}\label{xiuz}
\xi([u\inv],z) = \ds\sum_{i=1}^\infty \frac{u^{-i}}{i} \, z^i = -\log\left(1-\frac{z}{u}\right), 
\end{equation}
we get
	\begin{align*}
		\oint\dz & \Bigl(1-\frac{z}{u} \Bigr)\inv w(s,\t+[u\inv],\bt,\bx,z)w^*(s,\t,\bt,\bx,z) \\
		&= \oint\dz \w(s,\t+[u\inv], \bt,\bx,z)\w^*(s,\t,\bt,\bx,z).
	\end{align*}
Notice that $\w$ and $\w^*$ are formal power series of $z$, while $w-1$ and $w^*-1$ are formal power series of $z\inv$
(see \eqref{wavfn}).
Hence, the right-hand side of this equation vanishes, and
by Cauchy's formula, the left-hand side is
	\begin{equation*}
		u\Bigl(w(s,\t+[u\inv],\bt,\bx,u)w^*(s,\t,\bt,\bx,u) -1\Bigr) = 0.
	\end{equation*}
From this and \eqref{tau1}, we can derive \eqref{tau3}. 
Equation \eqref{tau4} is proved similarly.
\end{proof}

\begin{theorem}\label{betau}
A function\/ $\tau$ is a tau-function of the EBTH if and only if it satisfies
the following bilinear equation$:$
	\begin{equation}\label{beebth4}
	\begin{split}
		\oint & \ds\dz z^{nk+s'-s} e^{\xi(\t'-\t,z) -\frac12\xi_k(\t'-\t,z)} \\
		&\qquad\qquad\times \tau(s'-\xi(\ba,z^k), \t'-[z\inv], \bt', \bx + \ba) \\
		&\qquad\qquad\times \tau(s-\xi(\bb,z^k), \t+[z^{-1}],\bt,\bx+\bb) \\
		&=\oint\ds\dz z^{-nm+s'-s}e^{\xi(\bt-\bt',z^{-1}) - \frac12\xi_m(\bt-\bt',z^{-1})} \\
		&\qquad\qquad\times \tau(s'+1-\xi(\ba,z^{-m}), \t',\bt'+[z],\bx+\ba) \\
		&\qquad\qquad\times \tau(s-1-\xi(\bb,z^{-m}),\t,\bt-[z],\bx+\bb),
	\end{split}
	\end{equation}
for all\/ $\ba=(a_1,a_2,\dots)$, $\bb=(b_1,b_2,\dots)$, 
$n \in \Z_{\geq 0}$ and\/ $s-s' \in \Z$.
\end{theorem}
\begin{proof}
First, we plug in \eqref{beebth} the expressions for $\psi$, $\psibar$, $\psi^*$, $\psibar^*$ in terms of $\tau$
(see \eqref{tau1}, \eqref{tau2}, \eqref{tau3}, \eqref{tau4}).
Then, by \reref{zk}, we can replace $z^{nk}$ on the left-hand side of \eqref{beebth} by 
	\begin{equation*}
		z^{nk}\tau(s'-\xi(\ba,z^k), \t', \bt', \bx+\ba)\tau(s-\xi(\bb,z^k),\t,\bt, \bx+\bb),
	\end{equation*}
and $z^{-nm}$ on the right-hand side by 
	\begin{equation*}
		z^{-nm}\tau(s'-\xi(\ba,z^{-m}), \t',\bt',\bx+\ba)\tau(s-\xi(\bb,z^{-m}),\t,\bt,\bx+\bb).
	\end{equation*}
Therefore, \eqref{beebth} is equivalent to \eqref{beebth4}.
\end{proof}
If we apply the change of variables from \reref{rebth1}, we get the bilinear equation from \cite{CvdL13} (see $(85)$--$(87)$ there) as a special case of \eqref{beebth4} after setting $\bx = \boldsymbol 0$, 
$\ba = \bx'$, $\bb = \bx''$ in \eqref{beebth4}. Conversely, we can obtain \eqref{beebth4} from the bilinear equation of \cite{CvdL13} by observing that if $\tau(s,\t,\bt,\bx)$ is a tau-function for the EBTH, then so is $\tau(s,\t,\bt,\bx+\bm{c})$ for any constant $\bm{c}$. 


\subsection{Two difference Fay identities for the EBTH}


From \thref{betau}, we can derive the following difference Fay identities for the EBTH (cf.\ \cite{Tak08}).
We will again use the shift convention of \reref{rshift}.

	\begin{theorem}\label{thfay}
	If\/ $\tau$ is a tau-function of the EBTH, then for any\/ $\lambda, \mu \in \C^*$, 
	we have
			\begin{equation}\label{fay}
	\begin{split}
		\ds(\lambda&-\mu) \tau(s, \t, \bt, \bx) \tau(s-1,\t-[\lambda^{-1}] - [\mu^{-1}],\bt, \bx) \\
&= \lambda\,\tau(s,\t - [\lambda^{-1}],\bt,\bx)\tau(s-1, \t-[\mu^{-1}],\bt, \bx)\\ 
&- \mu\,\tau(s,\t - [\mu^{-1}],\bt, \bx)\tau(s-1, \t - [\lambda^{-1}],\bt,\bx)
	\end{split}
	\end{equation}
and
\begin{equation}\label{fay2}
\begin{split}
		(\lambda&-\mu) \tau(s+1,\t,\bt+[\lambda]+[\mu],\bx)\tau(s, \t,\bt,\bx) \\
		&= \lambda\,\tau(s+1,\t,\bt+[\lambda],\bx)\tau(s,\t,\bt+[\mu],\bx) \\ 
		&- \mu\,\tau(s+1,\t,\bt+[\mu],\bx)\tau(s,\t,\bt+[\lambda],\bx).
\end{split}
\end{equation}
	\end{theorem}
\begin{proof}
Using the same trick as in the proof of \prref{stars}, we can rewrite the bilinear equation \eqref{beebth4} as
		\begin{align*}
			\oint &\ds\dz z^{nk+s'-s} e^{\xi(\t'-\t,z)}
			\tau(s'-\xi(\ba,z^k), \t'-[z\inv], \bt', \bx+\ba) \\
			&\qquad\times\tau(s-\xi(\bb,z^k), \t+[z^{-1}],\bt, \bx+\bb) \\
			&=\oint\ds\dz z^{-nm+s'-s}\exp\Bigl({\ds-\sum_{m\nmid i} (\bar{t}'_i -\bar{t}_i)z^{-i}}\Bigr) \\
			&\qquad\times\tau(s'+1-\xi(\ba,z^{-m}), \t',\bt'+[z],\bx+\ba) \\
			&\qquad\times\tau(s-1-\xi(\bb,z^{-m}), \t,\bt-[z], \bx+\bb).
		\end{align*}
Then setting 
\begin{equation*}
n =0, \quad s' - s = 1, \quad \ba = \bb = \boldsymbol 0, \quad \t' = \t + [\lambda^{-1}] + [\mu^{-1}], \quad \bar{t}_{i}' = \bar{t}_{i},
\end{equation*}
for $m \nmid i$ gives
\begin{align*}
\oint &\ds\dz \ds\frac{z}{(1-z\lambda^{-1})(1-z\mu^{-1})}\\
&\qquad\quad\times\tau(s+1, \t + [\lambda^{-1}] + [\mu^{-1}] - [z^{-1}], \bt, \bx) \tau(s, \t + [z^{-1}], \bt, \bx)\\
&= \oint \ds\dz z\,\tau(s+2,\bt+[\lambda^{-1}] + [\mu^{-1}] + [z], \bt, \bx)\tau(s-1, \t - [z], \bt, \bx) \\
&= 0.
\end{align*}
To compute the residue in the left side, we use 
\begin{equation*}
\ds\frac{z}{(1-z\lambda\inv)(1-z\mu\inv)} = \ds\frac{1}{\lambda\inv - \mu\inv}\left(\ds\frac{1}{1-z\lambda\inv} - \ds\frac{1}{1-z\mu\inv}\right)
\end{equation*}
and 
\begin{equation*}
\oint\dz \ds\frac{f(z)}{1-z\lambda\inv} = \lambda(f(\lambda) - f_0), \qquad\text{if}\quad
f(z) = \ds\sum_{i=0}^\infty f_{i}z^{-i} .
\end{equation*}
We obtain
\begin{align*}
(\lambda &- \mu)\tau(s+1, \t + [\lambda^{-1}] + [\mu^{-1}],\bt, \bx)\tau(s,\t,\bt,\bx)\\
&-\lambda\,\tau(s+1,\t+[\lambda^{-1}],\bt, \bx)\tau(s, \t + [\mu^{-1}],\bt, \bx)\\ 
&+ \mu\,\tau(s+1, \t + [\mu^{-1}],\bt, \bx)\tau(s, \t + [\lambda^{-1}],\bt, \bx) = 0,
\end{align*}
which gives \eqref{fay} after the shift $s \mapsto s -1$, $\t \mapsto \t - [\lambda^{-1}] - [\mu^{-1}]$.
Equation \eqref{fay2} can be proved similarly, by making the substitution $s' = s - 1$, $\bt' = \bt - [\lambda] - [\mu]$ and $t_{i}' = t_{i}$ for $k \nmid i$ in \eqref{beebth4}.
\end{proof}


\section{Darboux transformations on the tau-function}\label{s3}

In this section, we first review the action of Darboux transformations on the Lax operator and wave function of the EBTH. Then, using the Fay identitity \eqref{fay}, we determine the action of Darboux transformations on the tau-function. 


\subsection{Darboux transformations of $L$ and $\psi$} 

A \emph{Darboux transformation} of a differential or difference operator $L$ is defined by factoring $L= QP$ and then switching the two factors: $L\mapsto L^{[1]} = PQ$ (see, e.g., \cite{AvM94, BHY296, GHZ, MS, MZ97}).
If $\phi$ is an eigenfunction of $L$ with $L\phi=\la\phi$, then $\phi^{[1]} = P\phi$ is an eigenfunction of $L^{[1]}$ with $L^{[1]}\phi^{[1]}=\la\phi^{[1]}$. If we repeat this process, the eigenfunction $\phi^{[N]}$ obtained after $N$ Darboux transformations can be given in terms of a Wronskian of the initial eigenfunction.

Darboux transformations for the ETH were first considered by G. Carlet in \cite{Car03}, 
and a generalization to the EBTH was given in \cite{LS16}. The following theorem is equivalent to Theorem 3.4 from \cite{LS16} and gives a formula for the wave function, $\psi^{[N]}$, and Lax operator, $L^{[N]}$, after $N$ iterations of the Darboux transformation. 
In order to state the theorem, we need to introduce some notation.

We will suppose that $\U\subset\C$ is an open set such that the wave function $\psi(s,\t,\bt,\x,z)$ is defined for $z\in\U$, i.e., the formal power series $w(z)$ from \eqref{wavfn} is convergent for $z\in\U$.
Then for $z_i\in\U$, we will denote $\psi_i=\psi|_{z=z_i}$.
We define the discrete Wronskian of functions $f_i = f_i(s)$ by
\begin{equation*}
		\Wr_\La(f_1, f_2, \dots, f_n) = \begin{vmatrix} f_1 & f_2 & \cdots & f_n \\ \La\inv (f_1) & \La\inv (f_2) & \cdots & \La\inv (f_n) \\ \vdots & \vdots & \cdots &\vdots \\ \La^{-n +1}(f_1) & \La^{-n+1}(f_2) & \cdots & \La^{-n+1}(f_n)\end{vmatrix} .
	\end{equation*}
	
	\begin{theorem}[\cite{LS16}]\label{DT}
	Let\/ $\psi$ be a wave function for the EBTH and\/ $L$ its corresponding Lax operator. For fixed\/ $N\ge 1$ and $z_1,\dots,z_N\in\U$, consider the difference operator\/ $P^{[N]}$ defined by
	\begin{equation*}
		P^{[N]}f = (-1)^N\ds\frac{\Wr_\La(\psi_1, \dots, \psi_N, f)}{\Wr_\La(\La\inv(\psi_1), \dots, \La\inv(\psi_N))} \,,
		\end{equation*}
where\/ $\psi_i=\psi|_{z=z_i}$. Then
		\begin{equation*}
		L^{[N]} = P^{[N]} L(P^{[N]})\inv, \qquad
		\psi^{[N]} = P^{[N]}\psi
		\end{equation*}
are a Lax operator and wave function for the EBTH, which are obtained from\/ $L$ and\/ $\psi$ after\/ $N$ Darboux transformations.
	\end{theorem}

To illustrate the theorem, consider the case of a single Darboux transformation. Then
\begin{equation}\label{psi1wr}
\begin{split}
\psi^{[1]} &= -\frac{\Wr_\La(\psi_1, \psi)}{\Wr_\La(\La\inv(\psi_1))} 
= -\frac{1}{\La\inv(\psi_1)} \begin{vmatrix} \psi_1 & \psi \\ \La\inv(\psi_1) & \La\inv(\psi) \end{vmatrix} \\
&= \psi - \frac{\psi_1}{\La\inv(\psi_1)} \La\inv(\psi), 
\qquad\text{where}\quad \psi_1=\psi|_{z=z_1}.
\end{split}
\end{equation}
Hence,
\begin{equation*}
P^{[1]} = I - \frac{\psi_1}{\La\inv(\psi_1)} \La\inv,
\end{equation*}
where $I$ denotes the identity operator (cf.\ \cite{Car03}).
Notice that $L\psi_1 = z_1 \psi_1$ and $P^{[1]} \psi_1 = 0$. Hence, by \cite[Theorem 2.3]{BGSZ15}, the difference operator $L-z_1 I$ 
factors as
\begin{equation*}
L-z_1 I = Q^{[1]} P^{[1]}
\end{equation*}
for some difference operator $Q^{[1]}$. Then the new Lax operator $L^{[1]}$ is obtained from the Darboux transformation
\begin{equation*}
L^{[1]}-z_1 I = P^{[1]} Q^{[1]},
\end{equation*}
and will have a wave function $\psi^{[1]}$. 

The fact that $L^{[1]}$ and $\psi^{[1]}$ are again solutions of the EBTH is one of the claims of \thref{DT} (see also \cite{Y19}).
The next Darboux transformation is done the same way, by starting from $L^{[1]}$, $\psi^{[1]}$ and $z_2$ in place of $L$, $\psi$ and $z_1$, respectively. The significance of \thref{DT} is that, after $N$ steps, the Lax operator $L^{[N]}$ and wave function $\psi^{[N]}$ can be expressed only in terms of the initial $L$ and $\psi$.
We refer to \cite{Y19} for more details and for a proof of \thref{DT} different from that of \cite{LS16}.

\subsection{Action of Darboux transformations on $\tau$} 

Using \thref{DT} and the Fay identity \eqref{fay}, we will prove that the action of a Darboux transformation on the tau-function is given by the \emph{vertex operator}
		\begin{equation}\label{vtx1}
		\Gamma_{+}(z) = e^{-\partial_s} e^{\xi(\t, z)} \exp\Bigl(-\ds\sum_{n=1}^\infty \ds\frac{\partial_{t_n}}{n}z^{-n}\Bigr).
		\end{equation}
Note that 
$\exp\bigl(-\sum_{n=1}^\infty \frac{\partial_{t_n}}{n}z^{-n}\bigr)$
acts as the shift operator $\t\mapsto\t-[z\inv]$,
while $e^{-\partial_s}=\La\inv$ acts as the shift $s\mapsto s-1$.

	\begin{theorem}\label{DTtau}
	Let\/ $\psi$ be a wave function for the EBTH, and\/ $\psi^{[1]}$ be the wave function after one Darboux transformation on\/ $\psi$
	$($see \eqref{psi1wr}$)$. Let\/ $\tau$ and\/ $\tau^{[1]}$ be their corresponding tau-functions. 
	Then\/ $\tau^{[1]}= \Gamma_+(z_1)\tau$, i.e.,
\begin{equation}\label{tau1DT}
\tau^{[1]}(s,\t,\bt,\bx) = e^{\xi(\t, z_1)} \tau(s-1,\t-[z_1\inv],\bt,\bx).
\end{equation}
	\end{theorem}

\begin{proof}
Using \eqref{psi1wr}, \eqref{tau1} and $\La\inv(\chi)=z\inv\chi$, we express $\psi^{[1]}$ in terms of $\tau$ as follows:
	\begin{equation}\label{psi1}
	\begin{split}
	\psi^{[1]} &=\ds\frac{\chi}{\tau(s,\t,\bt,\bx)\tau(s-1,\t-[z_1\inv],\bt,\bx)} \\[6pt]
	& \quad \times \Bigl(\tau(s,\t-[z\inv],\bt,\bx)\tau(s-1,\t-[z_1\inv],\bt,\bx)\\
	&\qquad\quad - z\inv z_1 \, \tau(s,\t-[z_1\inv],\bt,\bx)\tau(s-1,\t-[z\inv],\bt,\bx) \Bigr).
	\end{split}
	\end{equation}
On the other hand, 
again by \eqref{tau1}, 
	\begin{equation*}
	\psi^{[1]} = \ds\frac{\tau^{[1]}(s,\t-[z\inv],\bt,\bx)}{\tau^{[1]}(s,\t,\bt,\bx)}\chi.
	\end{equation*}
Substituting $\tau^{[1]} = \Gamma_+(z_1)\tau$ into the right side of this equation gives
	\begin{equation}\label{psi2}
	\ds\frac{(1-z\inv z_1) \tau(s-1,\t-[z\inv]-[z_1\inv],\bt,\bx)}{\tau(s-1,\t-[z_1\inv],\bt,\bx)} \chi,
	\end{equation}
where we used that, by \eqref{xiuz},
\begin{equation}\label{xiuz2}
e^{\xi(\t-[z\inv], z_1)} = e^{\xi(\t, z_1)} e^{-\xi([z\inv], z_1)} = e^{\xi(\t, z_1)} (1-z\inv z_1).
\end{equation}
If we set $\lambda = z$, $\mu = z_1$ in the Fay identity \eqref{fay}, 
we see that the above two expressions \eqref{psi1} and \eqref{psi2} are equal. Therefore, $\tau^{[1]}= \Gamma_+(z_1)\tau$.
\end{proof}

If we do $N$ Darboux transformations of $\tau$, we can apply \thref{DTtau} repeatedly to obtain the tau-function
\begin{equation}\label{tauN}
\tau^{[N]} = \Gamma_+(z_N) \cdots \Gamma_+(z_2) \Gamma_+(z_1) \tau,
\end{equation}
which corresponds to the Lax operator $L^{[N]}$ and wave function $\psi^{[N]}$ from \thref{DT}.
The product of vertex operators in \eqref{tauN} is well known (see, e.g., \cite[Chapter 14]{KacI}) 
and easy to compute using \eqref{tau1DT} and \eqref{xiuz2}. It follows that
	\begin{equation}\label{tauN2}
	\tau^{[N]}(s,\t,\bt,\bx) = 
	V_N
	e^{\sum_{i=1}^N\xi(\t, z_i)}
	\tau(s-N, \t-[z_1\inv]-\cdots-[z_N\inv], \bt, \bx),
	\end{equation}
where
\begin{equation}\label{VN}
V_N =\ds\prod_{1\leq i < j \leq N}\Bigl(1-\ds\frac{z_i}{z_j}\Bigr).
\end{equation}
One can verify directly that, for any tau-function $\tau$ of the EBTH, the function $\tau^{[N]}$ given by \eqref{tauN2} satisfies the bilinear equations \eqref{beebth4} and hence is a tau-function of the EBTH as well. 

\begin{remark}\label{minus} 
The authors of \cite{LS16} also give Darboux transformations on the second wave function, which is denoted $\psibar$ here. In this case, the action of the Darboux transformation on the tau-function is given by the vertex operator
\begin{equation}\label{vtx2}
\Gamma_-(z) = z^se^{\p{s}}e^{-\xi(\bt,z\inv)}\exp\Bigl(\sum_{n=1}^\infty\ds\frac{\partial_{\bar{t}_n}}{n}z^n\Bigr).
\end{equation}
The proof of this claim is very similar to the proof of \thref{DTtau} and uses \eqref{fay2} instead of \eqref{fay}; see \cite{Y19}. 
\end{remark}

As above, one can also use the bilinear equation \eqref{beebth4} to show directly that if $\tau$ is a tau-function for the EBTH, then $\Gamma_-(z_1)\tau$ is as well. We conclude that
\begin{equation*}
\Gamma_{\epsilon_N}(z_N)\cdots\Gamma_{\epsilon_1}(z_1)\tau
\end{equation*}
is a tau-function for the EBTH for any choice of signs $\epsilon_i = \pm$ (cf.\ \cite[Chapter 14]{KacI}).

\section{Generalized Fay identities}\label{s4}
In this section, as an application of Theorems \ref{DT} and \ref{DTtau}, we derive generalized difference Fay identities for the EBTH (see \cite{ASvM94} for the case of KP hierarchy). 
We will continue to use the notation of \seref{s3}.

\begin{theorem}\label{genfay}
Let\/ $\psi$ be a wave function for the EBTH with a corresponding tau-function\/ $\tau$, and let\/ $\psi_i = \psi\rvert_{z=z_i}$. Then 
		\begin{equation}\label{fay3}
		\begin{split}
		\ds\Wr_\La(\psi_1&, \ldots, \psi_{N}) = \chi_1\cdots\chi_{N} \prod_{1 \leq i < j \leq N} (z_j\inv - z_i\inv)
		\\
		&\times \frac{\tau(s-N+1,\t-[z_1\inv] - \cdots - [z_{N}\inv],\bt,\bx)}{\tau(s-N+1,\t,\bt,\bx)} \,,
		\end{split}
	\end{equation}
	where\/ 
	$\chi_i = \chi\rvert_{z=z_i}$.
	\end{theorem}
	
In this theorem, $z_1,\dots,z_N$ are complex numbers in a certain domain $\U\subset\C$, in which $\psi$ is defined. Alternatively, equation \eqref{fay3} makes sense as an identity of formal power series in $z_1\inv,\dots,z_N\inv$, if we write $\psi=w\chi$ for a formal power series $w$ in $z\inv$ (see \eqref{wavfn}), while the exponentials in $\chi$ are not expanded.

\begin{proof}[Proof of \thref{genfay}]
We will prove the claim by induction on $N$.
The case $N=1$ reduces to \eqref{tau1} for $z=z_1$, since $\Wr_\La(\psi_1) = \psi_1$. 
Now suppose that \eqref{fay3} holds for some $N\ge1$.

By \thref{DT}, we have
\begin{equation*}
\psi^{[N]} = (-1)^N\ds\frac{\Wr_\La(\psi_1, \dots, \psi_N, \psi)}{\Wr_\La(\La\inv(\psi_1), \dots, \La\inv(\psi_N))} \,.
\end{equation*}
After setting $z=z_{N+1}$, we obtain
\begin{equation*}
\psi^{[N]} \big|_{z=z_{N+1}}
 = (-1)^N\ds\frac{\Wr_\La(\psi_1, \dots, \psi_N, \psi_{N+1})}{\Wr_\La(\La\inv(\psi_1), \dots, \La\inv(\psi_N))} \,.
\end{equation*}
By the inductive assumption, the denominator is given by \eqref{fay3} after shifting $s\mapsto s-1$:
\begin{align*}
\Wr_\La(\La\inv(\psi_1)&, \dots, \La\inv(\psi_N)) 
= z_1\inv\cdots z_N\inv \chi_1\cdots\chi_{N} \prod_{1 \leq i < j \leq N} (z_j\inv - z_i\inv)
		\\
		&\times \frac{\tau(s-N,\t-[z_1\inv] - \cdots - [z_{N}\inv],\bt,\bx)}{\tau(s-N,\t,\bt,\bx)}\,.
\end{align*}
On the other hand, again by \eqref{tau1},
\begin{equation*}
\psi^{[N]}(s,\t,\bt,\bx,z) = \ds\frac{\tau^{[N]}(s,\t-[z\inv],\bt,\bx)}{\tau^{[N]}(s,\t,\bt,\bx)}\chi.
\end{equation*}
Let us plug here the formula \eqref{tauN2} for $\tau^{[N]}$ and set $z=z_{N+1}$. Using
\eqref{xiuz2} as before, we see that
\begin{align*}
\tau^{[N]}(s,\t&-[z_{N+1}\inv],\bt,\bx)
	= V_N \prod_{i=1}^N  (1- z_i z_{N+1}\inv) \, e^{\sum_{i=1}^N\xi(\t, z_i)} \\
	&\quad\times 
	\tau(s-N, \t-[z_1\inv]-\cdots-[z_{N+1}\inv], \bt, \bx).
\end{align*}
Hence,
\begin{align*}
\psi^{[N]} \big|_{z=z_{N+1}}
 &= \chi_{N+1} \prod_{i=1}^N  (1- z_i z_{N+1}\inv) \\
 &\quad\times
 \frac{\tau(s-N, \t-[z_1\inv]-\cdots-[z_{N+1}\inv], \bt, \bx)}{\tau(s-N, \t-[z_1\inv]-\cdots-[z_{N}\inv], \bt, \bx)} \,.
\end{align*}
Comparing the above two expressions for $\psi^{[N]} |_{z=z_{N+1}}$,
we obtain \eqref{fay3} with $N+1$ in place of $N$. This completes the proof of the theorem.
\end{proof}

Similarly, using Remark \ref{minus} and \cite[Theorem 5.3]{LS16}, we can obtain Fay identities with respect to $\bt$ given by
(see \cite{Y19}):
\begin{equation}\label{fay4}
		\begin{split}
		\Wr_\La^+(\psibar_1&, \ldots, \psibar_{N})
		= \chibar_1\cdots\chibar_{N} \prod_{1\leq i < j \leq N}\left(z_j-z_i\right) \\
		&\times\ds\frac{\tau(s+N,\t, \bt+[z_1] + \cdots + [z_{N}],\bx)}{\tau(s,\t,\bt,\bx)} \,,
	\end{split}
	\end{equation}
	where $\bar\psi_i = \bar\psi\rvert_{z=z_i}$, $\bar\chi_i = \bar\chi\rvert_{z=z_i}$ and
	\begin{equation*}
	\Wr_\La^+(f_1, f_2, \dots, f_n) = \begin{vmatrix} f_1 & f_2 & \cdots & f_n \\ \La (f_1) & \La (f_2) & \cdots & \La (f_n) \\ \vdots & \vdots & \cdots &\vdots \\ \La^{n-1}(f_1) & \La^{n-1}(f_2) & \cdots & \La^{n-1}(f_n)\end{vmatrix}.
	\end{equation*}

\section{Conclusion}\label{s5}

In this paper, we proved a bilinear equation for the extended bigraded Toda hierarchy (EBTH), which is equivalent to the bilinear equation of Carlet and van de Leur  \cite{CvdL13} after a change of variables but uses Takasaki's more convenient notation from \cite{Tak10}. From the bilinear equation, we derived difference Fay identities for the EBTH and showed that the action of the Darboux transformations on the wave functions $\psi$, $\psibar$ corresponds to acting on the tau-function by certain vertex operators $\Gamma_+$, $\Gamma_-$. 
As an application, we obtained generalized Fay identities for the EBTH. 

A natural question is to determine explicitly the initial tau-function corresponding to the trivial Lax operator $L = \Lambda^k + \Lambda^{-m}$, from which we can generate other solutions of the EBTH with Darboux transformations.
Wave functions for this Lax operator were given in \cite{Car03, LS16} in the cases $k = m = 1$ and $k = m = 2$, but they correspond to a wave function $\phi$ satisfying $L\phi = (z^k+z^{-m})\phi$, not $L\psi = z^k\psi$. We would like to determine the initial tau-function for the version of the EBTH presented here. 

Another interesting question is whether one can generate a $\mathcal{W}$-algebra from the vertex operators $\Gamma_+$ and $\Gamma_-$, as was done for the KP hierarchy in \cite{AvM92, AvM94,Di95}.  One can construct a Virasoro algebra based on \cite{BW16, DZ04}, but it would be interesting to try to construct a more general $\mathcal{W}$-algebra of symmetries by modifying the vertex operators $\Gamma_+$ and $\Gamma_-$ (cf.\ \cite{BM13, BS16, Mil07}).

We would also like to use our results about Darboux transformations to find solutions to the bispectral problem \cite{DG} for the EBTH (cf.\ \cite{BHY296, BHY596, BHY97}). The bispectral problem was first extended to difference operators in the case of the discrete KP hierarchy in \cite{HI00} and then expanded upon in \cite{GY02}.

%
%
%

\bibliographystyle{amsalpha}

\end{document}